\documentclass[conference]{IEEEtran}
\IEEEoverridecommandlockouts
\usepackage{textcomp}
\usepackage{csquotes}
\usepackage{microtype}
\usepackage{amsmath,amsthm,amssymb,amsfonts}
\usepackage{algorithm}
\usepackage{algpseudocode}
\usepackage[bookmarksdepth=3,linktoc=all,colorlinks=true,urlcolor=blue,linkcolor=blue,citecolor=blue]{hyperref}
\usepackage[capitalize]{cleveref}
\usepackage[style=ieee, sorting=nyt]{biblatex}
\usepackage{xcolor}
\usepackage{cancel}
\renewcommand{\theequation}{\arabic{section}.\arabic{equation}}

\newtheorem{assumption}{Assumption}
\newtheorem{theorem}{Theorem}

\theoremstyle{definition}
\newtheorem{definition}{Definition}
\theoremstyle{remark}

\def\ddefloop#1{\ifx\ddefloop#1\else\ddef{#1}\expandafter\ddefloop\fi}
\def\ddef#1{\expandafter\def\csname bb#1\endcsname{\ensuremath{\mathbb{#1}}}}
\ddefloop ABCDEFGHIJKLMNOPQRSTUVWXYZ\ddefloop
\def\ddef#1{\expandafter\def\csname bf#1\endcsname{\ensuremath{\mathbf{#1}}}}
\ddefloop ABCDEFGHIJKLMNOPQRSTUVWXYZabcdefghijklmnopqrstuvwxyz\ddefloop
\def\ddef#1{\expandafter\def\csname bs#1\endcsname{\ensuremath{\boldsymbol{#1}}}}
\ddefloop ABCDEFGHIJKLMNOPQRSTUVWXYZabcdefghijklmnopqrstuvwxyz\ddefloop
\def\ddef#1{\expandafter\def\csname sf#1\endcsname{\ensuremath{\mathsf{#1}}}}
\ddefloop ABCDEFGHIJKLMNOPQRSTUVWXYZ\ddefloop
\def\ddef#1{\expandafter\def\csname c#1\endcsname{\ensuremath{\mathcal{#1}}}}
\ddefloop ABCDEFGHIJKLMNOPQRSTUVWXYZ\ddefloop

\newcommand{\eps}{\epsilon}
\newcommand{\T}{\intercal}
\DeclareMathOperator*{\argmax}{arg\,max}
\DeclareMathOperator*{\argmin}{arg\,min}

\begin{document}

\title{\LARGE \bf
Distribution-Free Guarantees for Systems with Decision-Dependent Noise
}

\author{Heling Zhang$^{1}$, Lillian J. Ratliff$^{2}$ and Roy Dong$^{3}$
\thanks{$^{1}$Heling Zhang is with the Department of Electrical and Computer Engineering, University of Illinois at Urbana-Champaign,
        Champaign, IL 61801, USA
        {\tt\small hzhng120@illinois.edu}
        {\tt\small +1 217-898-9702}}%
\thanks{$^{2}$Lillian J. Ratliff is with the Department of Electrical and Computer Engineering, University of Washington,
        Seattle, WA 98195, USA
        {\tt\small ratliffl@uw.edu}
        {\tt\small +1 206-543-6244}}%
\thanks{$^{3}$Roy Dong is with the Department of Industrial and Enterprise Systems Engineering, University of Illinois at Urbana-Champaign,
        Champaign, IL 61801, USA
        {\tt\small roydong@illinois.edu}
        {\tt\small +1 217-265-4062}
        }%
}

\maketitle

\begin{abstract}
    In many real-world dynamicalsystems, obtaining precise models of system uncertainty remains a challenge. It may be difficult to estimate noise distributions or robustness bounds, especially when the distributions/robustness bounds vary with different control inputs in unknown ways. Addressing this challenge, this paper presents a novel iterative method tailored for systems with decision-dependent noise without prior knowledge of the distributions. Our approach finds the open-loop control law that minimizes the worst-case loss, given that the noise induced by this control lies in its $(1 - p)$-confidence set for a predetermined $p$. At each iteration, we use a quantile method inspired by conformal prediction to empirically estimate the confidence set shaped by the preceding control law. These derived confidence sets offer distribution-free guarantees on the system's noise, guiding a robust control formulation that targets worst-case loss minimization. Under specific regularity conditions, our method is shown to converge to a near-optimal open-loop control. While our focus is on open-loop controls, the adaptive, data-driven nature of our approach suggests its potential applicability across diverse scenarios and extensions.
\end{abstract}

\section{Introduction}
In the study of dynamical systems, the divergence between ideal models and real-world scenarios often manifests as a disturbance. This disturbance, compounded by the lack of knowledge about its statistical properties, presents significant challenges for control tasks. Traditional strategies, including stochastic and robust control, are theoretically effective but typically depend on prior knowledge or assumptions about the disturbance and its effect on the system, which are often unavailable in practice. Modern learning-based methods require less information about the distribution, but usually lack rigorous guarantees.

To address these challenges, we introduce an iterative method that gives distribution-free guarantees. That is, we assume no prior knowledge on the disturbance, even for the case where the disturbance is decision-dependent. The goal of our method is as follows: for a given noisy system with an unknown (potentially decision-dependent) disturbance, we aim to find the open-loop control law that minimizes the worst-case loss, given that the disturbance induced by this control lies in its $(1 - p)$-confidence set for a predetermined $p$.

Our method works as follows: at each iteration, we use a quantile method inspired by conformal prediction to empirically estimate the confidence set shaped by the preceding control law. These derived confidence sets offer empirical constraints on the system disturbance, guiding a robust control formulation that targets worst-case loss minimization. By solving this robust control problem, we obtain a refined control law which enables us to re-estimate the confidence set of the disturbance. Our theoretical analysis shows that this method converges to a near-optimal open-loop control law under specific regularity conditions.

Devising guarantees is a common theme in the study of noisy systems. When the noise is known, classic reachability analysis such as \cite{prajna07} provide provable guarantees of safety. When the noise distribution is unknown, various empirical methods are often used. In \cite{gillula13}, the authors derive safety guarantees by learning the worst-case disturbance online. In \cite{akametalu14}, the authors use a Gaussian process model to learn the system's unknown dynamics and approximates the maximal safe set. Our method adds to this list, providing a new way of deriving safety guarantees from empirical observations.

A key component of our method is inspired by \textit{conformal prediction}, a method widely used for uncertainty quantification in machine learning (\cite{vovk22, angelopoulos22}). In recent years, the application of conformal prediction in control tasks has also found great success. In \cite{lindemann23}, conformal prediction is used to construct predictors with probabilistic guarantees, which is then used to design model predictive controllers. In \cite{muthali23}, the authors uses conformal prediction to calibrate measures of uncertainty in trajectory forecasting models, which are then used to generate reachable sets that facilitates downstream navigation and planning tasks. In \cite{lin23}, the authors uses conformal prediction to verify neural reachable tubes derived with DeepReach (\cite{bansal20}). In \cite{lindemann22}, the authors provides two predictive runtime verification algorithms that use  conformal prediction to quantify prediction uncertainty and compute the probability that the system trajectory violates a signal temporal logic specification. In this work, we use  a quantile method inspired by conformal prediction to compute confidence sets of the noise induced by the current control law, which is then used to formulate a robust control problem that leads to a more refined control law.

Methodologically, our work is inspired by the works on \textit{performative prediction} \cite{perdomo20}, which is the study of decision-dependent distribution shift in learning. Conceptually, our problem shares several similarities with performative prediction: 1) the solutions to our problems alter the noise we encounter, and 2) we both seek the solution that gives the best performance in the noise it induces. As a result, our method adopts a similar iterative structure, in which each iteration builds upon the noise induced by the result from the previous iteration, but eventually converges to a fixed-point under certain regularity conditions.

\section{Problem Statement}
\label{sec_ps}
Consider a system that follows the underlying dynamics given by
\begin{equation}
    \label{eq_sys_under}
    x(t+1)
    =
    f(t, x(t), u(t), \zeta(t))
    \quad
    t = 0, 1, \dots, T-1,
\end{equation}
where $x(t) \in X_t \subseteq \bbR^n$ is the system state, $u(t) \in U_t \subseteq \bbR^m$ is the control input, and $\zeta(t) \in \bbR^n$ is a noise whose distribution depends on $x(t)$ and $u(t)$, i.e. $\zeta(t) \sim \cD(x(t), u(t))$. We assume that we do not have direct knowledge of the system dynamics, but instead, for any given sequence of control law $(u(0), \dots, u(T-1))$ and any initial state $x(0)$, we can take independent sample trajectories $(x_i(0), \dots, x_i(t))_{i=1}^{M}$. Out of these sample trajectories, we can then construct a unbiased estimator of $f(t, x(t), u(t), \zeta(t))$ using various methods. We write such unbiased estimator as $\hat{f}(t, x(t), u(t))$. Let $w(t) := f(t,x(t),u(t),\zeta(t)) - \hat{f}(t, x(t), u(t))$, we can see that $w(t)$ is a random variable whose distribution depends on $x(t)$ and $u(t)$. Then the system dynamics can be written as
\begin{equation}
    \label{eq_sys}
    x(t+1)
    =
    \hat{f}(t, x(t), u(t)) + w(t)
    \quad
    t = 0, 1, \dots, T-1.
\end{equation}
We are mainly concerned with developing open-loop control for the system described by Equation~\ref{eq_sys}, and assume $\hat{f}$ is a known function from this point on.

Without loss of generality, we further assume that the system starts at a fixed state $x(0) = x_0$. With a slight abuse of notation, we use $x, u, w$ to denote associated trajectories
\begin{align*}
    \mathbf{x} &= (x(0), \dots, x(T)) \in \bbR^{n(T+1)},\\
    \mathbf{u} &= (u(0), \dots, u(T-1)) \in \bbR^{mT},\\
    \mathbf{w} &= (w(0), \dots, w(T-1)) \in \bbR^{nT}.
\end{align*}
The objective is to determine an open-loop control law $\mathbf{u} = (u(0), ... u(T-1))$ that optimally governs the system's behavior. To quantify the performance of a control law, we introduce a loss function $J$ dependent on the entire trajectory of states and control inputs. We define $J$ as
\[
    \label{eq_loss}
    J(\mathbf{x}, \mathbf{u})
    :=
    \phi(x(T))
    +
    \sum_{t=0}^{T-1} 
    \varphi(x(t), u(t)).
\]
However, given the stochastic nature of the noise $w(t)$, the loss function $J$ becomes a random object. This makes it impossible to directly minimize the value of the loss function. Traditional approaches either aim to minimize the expected loss, requiring knowledge of the noise distribution, or minimize the worst-case loss for bounded noise, requiring knowledge of the extreme values of the noise. Both approaches require some prior knowledge of the behavior of the noise. This work addresses scenarios where no such prior knowledge is available.

When sampling is feasible, we can construct confidence sets for the noise with probabilistic guarantees using techniques similar to conformal prediction. The primary focus of this work is to minimize the worst-case loss given these confidence sets. Formally, for a designated probability level $p$, let $C_t(\mathbf{u})$ be the confidence set for $w(t)$ such that $w(t) \in C_t(\mathbf{u})$ with probability at least $1 - p/T$. Applying the union bound, this would give us $\mathbf{w} \in C_0(\mathbf{u}) \times \dots \times C_{T-1}(\mathbf{u})$ with probability at least $1-p$. We are interested in the following optimization problem
\begin{equation} \label{eq_p1}
\begin{aligned}
    \min_{\mathbf{u}}\quad 
    &\max_{\mathbf{w}} J(\mathbf{x}(\mathbf{u},\mathbf{w}), \mathbf{u}) \\
    \text{s.t}\quad 
    &\mathbf{w} \in C_0(\mathbf{u}) \times \dots \times C_{T-1}(\mathbf{u}) \\
    &\mathbf{u} \in U_0 \times \dots \times U_{T-1}.
\end{aligned}
\end{equation}
For clarity, we can also equivalently express the loss as $L(\mathbf{w}, \mathbf{u}) = J(\mathbf{x}(\mathbf{u},\mathbf{w}), \mathbf{u})$, emphasizing its dependence on the noise and control. Further, let 
\begin{align*}
    C(\mathbf{u}) &:= C_0(\mathbf{u}) \times \dots \times C_{T-1}(\mathbf{u})  \\
    U &:= U_0 \times \dots \times U_{T-1},
\end{align*}
then we can write the optimization problem in a more concise form
\begin{equation*}\label{eq_p2}
\begin{aligned} 
    \min_{\mathbf{u}} \quad &\max_{\mathbf{w} \in C(\mathbf{u})} L(\mathbf{w}, \mathbf{u})\\
    \text{s.t}\quad &\mathbf{u} \in U.
\end{aligned}
\end{equation*}

\section{Methods}
\label{sec_m}
In this section, we introduce our main method, empirical iterative refinement of performative control (E-IRPC). We call the control performative since it induces changes in the system noise distribution.

\subsection{Empirical Iterative Refinement of Performative Control}
The core of our methodology is an iterative process designed to progressively refine the control law to achieve optimal performance in the presence of noise. Our method consists of the following steps:
\begin{enumerate}
    \item
        \textbf{Initialization:}
        We start by finding the optimal open-loop control 
        for the nominal system (i.e. the system without noise)
        with loss function $J$:
        \begin{equation} \label{eq_init}
        \begin{aligned}
            \min_{\mathbf{u}}\quad 
            &J(\mathbf{x}(\mathbf{u}), \mathbf{u}) \\
            \text{s.t}\quad 
            &x(t+1) = \hat{f}(t, x(t), u(t)) \\
            &\mathbf{u} \in U_0 \times \dots \times U_{T-1}.
        \end{aligned}
        \end{equation}
        We denote the resulting control law as 
        $\mathbf{u}_0 = (u_0(0), \dots, u_0(T-1))$.
    \item
        \textbf{Sampling and Confidence Set Construction:}
        In this step, we fix the control law and compute
        the confidence sets for the induced noise.
        Specifically,
        with the control law from the previous step, we sample 
        ${N_1}$ independent trajectories of system states, represented as
        $\mathbf{x}^{(j)}_1$ for $j = 1, \dots, {N_1}$. Since $\hat{f}$ is known, this gives us 
        ${N_1}$ trajectories of the noise, $\mathbf{w}^{(j)}_1$ for $j = 1, \dots, {N_1}$.
        Notice that since we are using the same control input,
        at each time step $t$, $(w^{(j)}_1(t))_{j=1}^{N_1}$ are i.i.d.
        We can then use quantile regression to construct sets $\hat{C}_t(\mathbf{u})$ such that 
            $\bbP(w(t) \in \hat{C}_t(\mathbf{u})) 
            \ge
            1 - \frac{p}{T}$.
        Again, for simplicity, we denote
        $\hat{C}(\mathbf{u}) := \hat{C}_0(\mathbf{u}) \times \dots \times \hat{C}_{T-1}(\mathbf{u})$.
        The details of such construction will be outlined in the 
        next subsection.
    \item
        \textbf{Robust Control Problem Formulation:}
        In this step, we fix the confidence sets and
        vary the control law. Specifically, 
        with the confidence sets computed in the previous step, 
        we can formulate a robust control problem similar to 
        Equation~\ref{eq_p1}:
        \begin{equation}
            \label{eq_rc}
        \begin{aligned}
            \min_{\mathbf{u}}\quad 
            &\max_{\mathbf{w}} J(\mathbf{x}(\mathbf{u},\mathbf{w}), \mathbf{u}) \\
            \text{s.t}\quad 
            &\mathbf{w} \in \hat{C}_0(u_0) \times \dots \times \hat{C}_{T-1}(u_0) \\
            &\mathbf{u} \in U_0 \times \dots \times U_{T-1}.
        \end{aligned}
        \end{equation}
        We solve this problem to obtain a new control law $\mathbf{u}_1$.
    \item
        \textbf{Iterative Refinement:}
        The new control law is then used as the initial control law
        in Step 2, and the process repeats.
\end{enumerate}
In summary, for $i = 0, 1, ...$, we fix the control law $\mathbf{u}_i$, sample $N_i$ trajectories $(\mathbf{x}^{(j)}_{i+1})_{j = 1}^{N_i}$ and $(\mathbf{w}^{(j)}_{i+1})_{j = 1}^{N_i}$ to construct a confidence set $\hat{C}(\mathbf{u}_i)$, then we fix $\hat{C}(\mathbf{u}_i)$ and formulate a robust control problem to obtain $u_{i+1}$. To simplify our notation, define $\hat{\cA}: U \to U$ as
\begin{equation}
    \label{eq_ir}
    \hat{\cA}(\mathbf{u})
    \in 
    \argmin
    \left \{
        \max
        \{
            L(\mathbf{w}, \mathbf{v})
            :
            \mathbf{w} \in \hat{C}(\mathbf{u})
        \}
        :
        \mathbf{v} \in U
        \right \},
\end{equation}
then our method is summarized in Algorithm~\ref{alg_iter}.

The iterative nature of this method is essential due to the performative aspect of the control: the control law influences the distribution of the noise, necessitating repeated refinement. Additionally, without prior knowledge of how different controls influence this noise, we need to re-collect samples to estimate bounds on the noise whenever we change the control.

\begin{algorithm}
\caption{E-IRPC}\label{alg_iter}
\begin{algorithmic}
\State $u_0 \gets \text{Solution to Problem~\ref{eq_init}}$
\While{$i > 0$}
    \State Sample $(\mathbf{x}^{(j)}_{i+1})_{j = 1}^{N_i}$ and $(\mathbf{w}^{(j)}_{i+1})_{j = 1}^{N_i}$
    \State Construct $\hat{C}(\mathbf{u}_{i-1})$ from samples
    \State $\mathbf{u}_i \gets \hat{\cA}(\mathbf{u}_{i-1})$  \Comment{$\hat{\cA}$ defined in Equation~\ref{eq_ir}.}
\EndWhile
\end{algorithmic}
\end{algorithm}

\subsection{Confidence Set Construction}
\label{sec_cs}
In this section, we will outline the method we use to construct the confidence sets in Step 2. As described in the previous section, for each time step $t$, we have ${N_i}$ samples of $w(t)$. Let's sort their Euclidean norm in the increasing order as follows: $\left |w^{(\eta_1)}(t) \right | 
    \le
    \dots
    \le
    \left |w^{(\eta_{N_i})}(t) \right |$. 
Since these samples are i.i.d, we can easily show that the Euclidean norm of another i.i.d sample $w(t)$ is equally likely to fall between that of any existing samples \cite{angelopoulos22}, this gives us
\begin{equation}
    \label{eq_cp}
    \bbP\left(
        |w(t)|
        \le
        |w^{(\eta_k)}(t)|
    \right)
    =
    \frac{k}{N_i+1}
\end{equation}
for any $k = 1, \dots, N_i$. This allows us to construct the confidence sets with desired probabilistic guarantee. Specifically, for the overall probability level $p$, pick
$
    k 
    = 
    \left\lceil
        (N_i + 1)\left(1 - \frac{p}{T} \right)
    \right\rceil
$
and define
\begin{equation}
    \label{eq_cs}
    \hat{C}_t(\mathbf{u})
    :=
    \left\{
        w(t)
        :
        |w(t)|
        \le
        |w^{(\eta_k)}|
    \right\},
\end{equation}
by Equation~\ref{eq_cp}, we have
\begin{equation*}
    \bbP\left(
        w(t) 
      \mathbf{ } \in
        \hat{C}_t(\mathbf{u})
    \right)
    =
    \frac{
        \left\lceil
            (N_i + 1)\left(1 - \frac{p}{T} \right)
        \right\rceil
    }{
        N_i + 1
    }
    \ge
    1 - \frac{p}{T},
\end{equation*}
which means that Equation~\ref{eq_cs} gives us the desired confidence set. It is worth noting that when constructing these confidence sets, we can replace the Euclidean norm with any other score function. This will allow us to construct confidence sets with the same probabilistic guarantee but different shapes. For example, if we can chose $s(w(t)) := \sqrt{w^\T(t) H w(t)}$ for some positive definite definite $H$, and rank the sample noises according to the score
$
    s\left(
        w^{(\eta_1)}(t)
    \right)
    \le
    \dots
    \le
    s\left(
        w^{(\eta_{N_i})}(t)
    \right)
$, 
the resulting confidence set given by
$
    \tilde{C}_t (\mathbf{u}) 
    = 
    \left\{
        w(t) : 
        s(w(t)) 
        \le
        s\left(
            w^{(\eta_k)}(t)
        \right)
    \right\}
$
will still have the same probabilistic guarantee, but takes the shape of an ellipsoid instead of a ball.

\subsection{Ideal Iterative Refinement of Performative Control}
So far, we have detailed our primary method, which inherently relies on estimating specific properties of the noise distribution using finite samples. Naturally, the algorithm's performance is influenced by the number of samples taken. In this section, we introduce an idealized variant of our method, termed the \textit{Ideal Iterative Refinement of Performative Control} (I-IRPC). This version sidesteps the variability introduced by finite sample trajectories. As we will discuss, the I-IRPC can be conceptualized as the infinite sample counterpart of our E-IRPC.

Formally, let $\mathbf{u}$ be the control input for the system defined by \ref{eq_sys}, and $w(t)$ be the noise at time $t$ whose distribution is given by $\cD(x(t), u(t))$. We define the \textit{ideal confidence set} for $w(t)$ as
\begin{equation}
    \label{eq_ics}
    C_t(\mathbf{u}) = \{w(t) : |w(t)| \le Q_t(\mathbf{u})\},
\end{equation}
where $Q_t(\mathbf{u})$ is the quantile function of $|w(t)|$ defined as
\begin{equation}
    \label{eq_qf}
    Q_t(\mathbf{u})
    :=
  \mathbf{ } \inf
    \left \{
        r
        :
        \bbP(|w(t)| \le r)
        \ge
        1 - 
        \frac{p}{T}
    \right \}
\end{equation}
where $p$ is the overall probability level. The set $C_t(\mathbf{u})$ is the idealized counterpart of $\hat{C}_t(\mathbf{u})$ defined in \ref{eq_cs}, which is impossible to obtain without knowing the exact distribution of $w(t)$. In the I-IRPC, we simply assume that there is an oracle that gives $C_t(\mathbf{u})$ for any $\mathbf{u}$ we choose.

Notice that $Q_t$ depends solely on $\mathbf{u}$ because $\mathbf{u}$ uniquely defines the distribution of $w(t)$. Now, let $C(\mathbf{u}) := C_0(\mathbf{u}) \times \dots \times C_{T-1}(\mathbf{u})$ and define $\cA: U \to U$ as
\begin{equation}
    \label{eq_iir}
    \cA(\mathbf{u}) \in 
    \argmin
    \left\{
        \max\{
            L(\mathbf{w}, \mathbf{u})
            :
            \mathbf{w} \in C(\mathbf{u})
        \}
        :
        \mathbf{u} \in U
    \right\}.
\end{equation}
Then the I-IRPC is given by Algorithm~\ref{alg_iir}.
\begin{algorithm}
\caption{I-IRPC}\label{alg_iir}
\begin{algorithmic}
\State $u_0 \gets \text{Solution to Problem~\ref{eq_init}}$
\While{$i > 0$}
    \State Get $C(\mathbf{u})$ from the oracle
    \State $\mathbf{u}_i \gets \cA(\mathbf{u}_{i-1})$  \Comment{$\cA$ defined in Equation~\ref{eq_iir}.}
\EndWhile
\end{algorithmic}
\end{algorithm}
The I-IRPC gives us a starting point for analysis. Specifically, due to its deterministic nature, it allows us to make the following key definitions.
\begin{definition}[Performatively Stable Control]
    \label{def_ps}
       We say $\mathbf{u}_{PS} \in U$ is a \textit{performatively stable control} if
    $
        \cA(\mathbf{u}_{PS}) = \mathbf{u}_{PS}
    $.
\end{definition}
\begin{definition}[Performatively Optimal Control]
    \label{def_po}
    We say $\mathbf{u}_{PO} \in U$ is a \textit{performatively optimal control} if 
    \[
        \mathbf{u}_{PO}
      \mathbf{ } \in
        \argmin_{\mathbf{u} \in U}
        \max_{\mathbf{w} \in C(\mathbf{u})}
        L(\mathbf{w}, \mathbf{u}),
    \]
    where $C(\mathbf{u})$ is defined in Equation~\ref{eq_ics}.
\end{definition}
Conceptually, the performatively stable control is the control law that performs optimally in the noise it induces. The performatively optimal control, on the other hand, gives the best performance among all feasible control laws in the presence of self-induced noise. Notice that the performative optimal control, in contrast to performative stable control, may not be the optimal control law in the noise that it induces. As we will see shortly, these concepts play a central role in our theoretical analysis of our proposed method.

Now, let's revisit E-IRPC from the lens of I-IRPC. Define the empirical quantile function by
\begin{equation*}
\begin{aligned}
    \label{eq_eqf}
    \hat{Q}_t(\mathbf{u}) 
    &:= 
    \inf
    \left \{
        r
        :
        \frac{1}{{N_i}}
        \sum_{j = 1}^{N_i}
        \mathbf{1}\left\{
            \left | w^{(j)}(t) \right |
            \le
            r
            \right\}
        \ge
        1 - \frac{p}{T}
        \right \} \\
    &=
    \left | w^{(\eta_k)}(t) \right |.
\end{aligned}
\end{equation*}
As the name suggests, the empirical quantile function is a finite sample approximation of the actual quantile function defined in Equation~\ref{eq_qf}. The confidence sets constructed in Section~\ref{sec_cs} is then given by
\begin{equation*}
\begin{aligned}
    &\hat{C}_t(\mathbf{u})
    =
    \left \{
        w(t) : |w(t)| \le \hat{Q}_t(\mathbf{u})
    \right \},
    \\
    &\hat{C}_u
    =
    \hat{C}_0(\mathbf{u})
    \times
    \dots
    \times
    \hat{C}_{T-1}(\mathbf{u}).
\end{aligned} 
\end{equation*}
In this sense, the E-IRPC algorithm is a finite-sample approximation of I-IRPC. This connection is essential to the upcoming analysis.

\section{Main Results}
\label{sec_mr}
In this section, we provide a theoretical analysis of our proposed method. We start with an examination of the I-IRPC, where we give sufficient conditions that guarantees 1) the existence and uniqueness of the performatively stable control, 2) convergence to the performatively stable control and 3) the proximity of the performatively stable control to the performatively optimal control. We then extend these results to the E-IRPC.

\subsection{Preliminaries}
We start by presenting some standard assumptions on the loss function $L(\mathbf{w}, \mathbf{u})$ and an important assumption regarding the noise distribution.
\begin{assumption}[$\lambda$-strong convexity in $\mathbf{u}$]
    \label{asm_sc}
    We assume that the loss function $L(\mathbf{w}, \mathbf{u})$ is $\lambda$-strongly convex
    in $\mathbf{u}$.
\end{assumption}
\begin{assumption}[$\beta$-smoothness in $\mathbf{w}$]
    \label{asm_sm}
    We assume that the loss function $L(\mathbf{w}, \mathbf{u})$ is $\beta$-smooth in $\mathbf{w}$,
\end{assumption}
\begin{assumption}[Lipschitz continuity of quantile functions]
    \label{asm_lip}
    Let $Q_t(\mathbf{u})$ be the $(1 - p/T)$-quantile function for 
    $|w(t)|$ when $\mathbf{u}$ is the control input and $p$
    is the predefined probability level 
    (the definition of
    $Q_t(\mathbf{u})$ is given in Equation~\ref{eq_qf}).
    We assume that $Q_t(\mathbf{u})$ is $\eps_t$-Lipschitz, i.e.
    for all $\mathbf{u}, \mathbf{u}' \in U$,
    $
        |Q_t(\mathbf{u}) - Q_t(\mathbf{u}')|
        \le
        \eps_t
        |\mathbf{u} - \mathbf{u}'|
    $.
\end{assumption}

\subsection{Existence of fixed-points and optima}
Let
$
    g(\mathbf{u}, \mathbf{v})
    =
    \max\{
        L(\mathbf{w}, \mathbf{v})
        :
        \mathbf{w} \in C(\mathbf{u})
    \}
$.
\begin{theorem}[Existence of performatively stable control]
    \label{thm_eps}
    Suppose the following assumptions hold:
    \begin{itemize}
        \item
            $L(\mathbf{w}, \mathbf{u})$ is strictly convex in $\mathbf{u}$,
        \item
            $g(\mathbf{u}, \mathbf{v})$ jointly continuous in $\mathbf{u}$ and $\mathbf{v}$,
        \item
            the quantile functions $Q_t(\mathbf{u})$ are 
            $\eps_t$ in $\mathbf{u}$ for all $t = 0, \dots, T-1$ (Assumption~\ref{asm_lip}),
            and
        \item
            $U$ is compact and convex.
    \end{itemize}
    Then there exists at least one performatively stable control.
\end{theorem}
\begin{proof}
    Since $g(\mathbf{u}, \mathbf{v})$ jointly continuous,
    by the Maximum Theorem (see, e.g. \cite[Chapter~E.3]{ok07}),
    we know that the set valued mapping
    given by 
    $
        \mathbf{u} \mapsto \argmin\{g(\mathbf{u}, \mathbf{v}): \mathbf{v} \in U\}
    $
    is upper hemicontinuous. Further, by strict convexity of $L(\mathbf{w}, \mathbf{u})$
    in $\mathbf{u}$, $g(\mathbf{u}, \mathbf{v})$ is also strictly convex in $\mathbf{u}$. 
    Therefore, the set $\argmin\{g(\mathbf{u}, \mathbf{v}): \mathbf{v} \in U\}$ contains only 
    one element, which is $\cA(\mathbf{u})$.
    This implies that $\cA(\mathbf{u})$ is continuous.
    The existence of a fixed points then follows
    by the Brouwer's Fixed-Point Theorem.
\end{proof}

Uniqueness of the performatively stable control requires a bit more conditions, as we will see shortly. The existence and uniqueness of the performatively optimal control, on the other hand, is much easier to establish.
\begin{theorem}[Existence and Uniqueness of performatively optimal control]
    \label{thm_epo}
    Suppose the following assumptions hold:
    \begin{itemize}
        \item 
            $L(\mathbf{w}, \mathbf{u})$ is jointly continuous in $\mathbf{u}$ and
        \item
            $U$ is compact and convex.
    \end{itemize}
    Then there exists a unique performatively optimal control.
\end{theorem}
\begin{proof}
    Assumption~\ref{asm_lip} implies that the 
    set-valued map
    $\mathbf{u} \mapsto C(\mathbf{u})$
    is continuous.
    Define $h(\mathbf{u}) 
        := 
        \max_{\mathbf{w} \in C(\mathbf{u})} L(\mathbf{w}, \mathbf{u})
        =
        g(\mathbf{u}, \mathbf{u})$. 
    Notice the performative optimal control is the minimizer 
    of $h(\mathbf{u})$ over $U$.
    Since $C(\mathbf{u})$ is continuous, 
    by the Maximum Theorem,
    $h(\mathbf{u})$ is continuous in $\mathbf{u}$.
    The existence of the performative optimal control then
    follows from the Weierstrass Extreme Value Theorem.
\end{proof}

\subsection{Convergence}
We are now ready to investigate the convergence properties of the I-IRPC algorithm. Our analysis relies on the following additional assumption.

\begin{definition}
    We say that the worst case noise align
    for the loss function $L$ and the confidence set mapping $C$
    if for any $\mathbf{v}, \mathbf{u}, \mathbf{u}' \in U$, align if there exists
    $\mathbf{w}_1^* \in \argmax_{\mathbf{w} \in C(\mathbf{u})} L(\mathbf{w}, \mathbf{v})$ and
    $\mathbf{w}_2^* \in \argmax_{\mathbf{w} \in C(\mathbf{u}')} L(\mathbf{w}, \mathbf{v})$
    such that $\mathbf{w}_1^*/|\mathbf{w}_1^*| = \mathbf{w}_2^*/|\mathbf{w}_2^*|$.
\end{definition}
\begin{assumption}[Alignment of worst-case noises]
    \label{asm_noise1}
    We assume that the worst-case noises \textit{align}.
\end{assumption}

\begin{theorem}[Convergence of I-IRPC to $\mathbf{u}_{PS}$]
    \label{thm_cvg1}
    Under Assumptions~\ref{asm_sc},~\ref{asm_sm},~\ref{asm_lip}, and~\ref{asm_noise1}, 
    the I-IRPC converges
    to a unique performative 
    stable control $\mathbf{u}_{PS}$ at a linear rate: 
        $|\mathbf{u}_i - \mathbf{u}_{PS}| \le \delta$ 
    for all
    \begin{equation*}
        i 
        \ge
        \left ( 
            1 
            - 
            \frac{
                \beta
                \sqrt{
                    \sum_{t=0}^{T-1}
                    \eps_t^2
                }
            }{\lambda}
        \right )^{-1}
        \log
        \left (
            \frac
            {|\mathbf{u}_0 - \mathbf{u}_{PS}|}
            {\delta}
        \right ).
    \end{equation*}
\end{theorem}
\begin{proof}
    Let $\mathbf{u}, \mathbf{u}' \in U$ be two different control inputs.
    Since $L(\mathbf{w}, \mathbf{u})$ is $\lambda$-strongly convex in $\mathbf{u}$,
    for any $\mathbf{u} \in U$, $g(\mathbf{u}, \mathbf{v})$ 
    is $\lambda$-strongly convex in $\mathbf{u}$.
    Let $\partial_\mathbf{u} g(\mathbf{u}, \mathbf{v})$ denote the subgradient of $g(\mathbf{u}, \mathbf{v})$
    with respect to $\mathbf{u}$.
    Take any
    $s_\mathbf{u} \in \partial_\mathbf{{}\cA(\mathbf{u})} g(\mathbf{u}, \cA(\mathbf{u}))$
    and     
    $s_{\mathbf{u}'} \in \partial_\mathbf{{}\cA(\mathbf{u}')} g(\mathbf{u}, \cA(\mathbf{u}'))$,
    by strong convexity and definition of subgradients, 
    we can derive
        $\lambda
        |\cA(\mathbf{u}) - \cA(\mathbf{u}')|^2
        \le
        (\cA(\mathbf{u}') - \cA(\mathbf{u}))^\T 
        (s_{\mathbf{u}'} - s'_{\mathbf{u}'})$.
    Notice again that this hold for any 
    $s_{\mathbf{u}'} \in \partial_\mathbf{{}\cA(\mathbf{u}')} g(\mathbf{u}, \cA(\mathbf{u}'))$ and any
    $s'_{\mathbf{u}'} \in \partial_\mathbf{{}\cA(\mathbf{u}')} g(\mathbf{u}', \cA(\mathbf{u}'))$.
    Now, take $\mathbf{w}_1^* \in \argmax_{\mathbf{w} \in C(\mathbf{u})} L(\mathbf{w}, \cA(\mathbf{u}'))$ and $\mathbf{w}_2^* \in \argmax_{\mathbf{w} \in C(\mathbf{u}')} L(\mathbf{w}, \cA(\mathbf{u}'))$.
    It's easy to show that
        $\nabla_{\cA(\mathbf{u}')}
        L(\mathbf{w}_1^*, \cA(\mathbf{u}'))
        \in 
    \partial_\mathbf{{}\cA(\mathbf{u}')} g(\mathbf{u}, \cA(\mathbf{u}'))$ and 
        $\nabla_{\cA(\mathbf{u}')}
        L(\mathbf{w}_2^*, \cA(\mathbf{u}'))
        \in 
    \partial_\mathbf{{}\cA(\mathbf{u}')} g(\mathbf{u}', \cA(\mathbf{u}'))$.
    Combining these, we have
    \begin{equation*}
    \begin{aligned}
        &\qquad|\cA(\mathbf{u}) - \cA(\mathbf{u}')|\\
        &\le
        \frac{1}{\lambda}
        \left |
        \nabla_{\cA(\mathbf{u}')}
        L(\mathbf{w}_1^*, \cA(\mathbf{u}'))
        -
        \nabla_{\cA(\mathbf{u}')}
        L(\mathbf{w}_2^*, \cA(\mathbf{u}'))
        \right |
        \\
        &\stackrel{(i)}{\le}
        \frac{\beta}{\lambda}
        |\mathbf{w}_1^* - \mathbf{w}_2^*|
        \\
        &\stackrel{(ii)}{=}
        \frac{\beta}{\lambda}
        \sqrt{
            \sum_{t=0}^{T-1}
            |Q_t(\mathbf{u}) - Q_t(\mathbf{u}')|^2
        }
        \\
        &\stackrel{(iii)}{\le}
        \frac{
            \beta
            \sqrt{
                \sum_{t=0}^{T-1}
                \eps_t^2
            }
        }{\lambda}
        |\mathbf{u} - \mathbf{u}'|.
    \end{aligned}
    \end{equation*}
    where step (i) follows from 
    $\beta$-smoothness (Assumption~\ref{asm_sm}),
    step (ii) follows from 
    Assumption~\ref{asm_noise1}
    and step (iii) follows from
    $\eps_t$-Lipschitz continuity (Assumption~\ref{asm_lip}).
    Thus, when
    $\frac{
            \beta
            \sqrt{
                \sum_{t=0}^{T-1}
                \eps_t^2
            }
        }{\lambda}
        <
        1$, 
    the I-IRPC converges as stated by the
    Contraction Mapping Theorem 
    (see, for example, \cite[Appendix~B]{khalil02}).
\end{proof}

\subsection{Relating Performatively Stable Control to Performatively Optimal Control}
So far, we have identified the conditions that allow the I-IRPC to converges to a unique performatively stable control, $\mathbf{u}_{PS}$. However, what we really care about is the performatively optimal control, $\mathbf{u}_{PO}$, which is the solution to Problem~\ref{eq_p1}. In this section, we will establish the sufficient condition that makes $\mathbf{u}_{PO}$ close to $\mathbf{u}_{PS}$. To do so, we first need to make the following additional assumption.
\begin{assumption}[$L_w$-Lipschitz continuity in $\mathbf{w}$]
    \label{asm_lip2} 
    We assume that $L(\mathbf{w}, \mathbf{u})$ is $L_w$-Lipschitz in $\mathbf{w}$, i.e.
    for any $\mathbf{w}, \mathbf{w}' \in \bbR^{nT}$ and $\mathbf{u} \in U$,
$|L(\mathbf{w}, \mathbf{u}) - L(\mathbf{w}', \mathbf{u})|
        \le
        L_w |\mathbf{w} - \mathbf{w}'|$. 
\end{assumption}
\begin{theorem}[Relating $\mathbf{u}_{PS}$ and $\mathbf{u}_{PO}$]
    \label{thm_close1}
    In addition to assumptions made in Theorem~\ref{thm_cvg1},
    if Assumption~\ref{asm_lip2} is also satisfied, then 
    \[
    |\mathbf{u}_{PS} - \mathbf{u}_{PO}|
        \le
        \frac{
            2 L_w
            \sqrt{
                \sum_{t=0}^{T-1}
                \eps_t^2
            }
        }{\lambda}.
        \]
\end{theorem}
\begin{proof}
    By the definition of $\mathbf{u}_{PO}$ and $\mathbf{u}_{PS}$, we have $g(\mathbf{u}_{PO}, \mathbf{u}_{PO})
        \le
        g(\mathbf{u}_{PS}, \mathbf{u}_{PS})
        \le
        g(\mathbf{u}_{PS}, \mathbf{u}_{PO})$. 
    Since $\mathbf{u}_{PS}$ is a maximizer of $g(\mathbf{u}_{PS}, \mathbf{v})$, we know
    that $0 \in \partial_\mathbf{v} g(\mathbf{u}_{PS}, \mathbf{v})$. Therefore, by
    $\lambda$-strong convexity (Assumption~\ref{asm_sc}), we have
        $g(\mathbf{u}_{PS}, \mathbf{u}_{PO})
        -
        g(\mathbf{u}_{PS}, \mathbf{u}_{PS})
        \ge
        \frac{\lambda}{2}
        |\mathbf{u}_{PS} - \mathbf{u}_{PO}|^2$. 
    Furthermore,
    by Assumption~\ref{asm_lip2}, we also have
        $g(\mathbf{u}_{PS}, \mathbf{u}_{PO})
        -
        g(\mathbf{u}_{PO}, \mathbf{u}_{PO})
        \le
        L_w
        |\mathbf{w}_1^* - \mathbf{w}_2^*|$.
    Combining above equations and applying Assumption~\ref{asm_noise1},
    we get
    \begin{equation*}
        |\mathbf{u}_{PS} - \mathbf{u}_{PO}|^2
        \le
        \left (
            \frac{2 L_w}{\lambda}
            \sqrt{
                \sum_{t=0}^{T-1}
                \eps_t^2
            }
        \right )
        |\mathbf{u}_{PS} - \mathbf{u}_{PO}|.
    \end{equation*}
    Simplifying this completes the proof.
\end{proof}

\subsection{Finite Sample Results}
So far, we have focused our analysis on the I-IRPC, which is the ideal version of our proposed method. Now we are finally ready to bridge the gap between them, and establish the convergence properties of our E-IRPC algorithm. The key here is to view the radius of confidence sets obtained in Section~\ref{sec_cs} as the empirical approximation of the quantile function defined by Equation~\ref{eq_qf}.

\begin{definition}[Hazard rate]
    \label{def_hr}
    Let $X$ be a $1$-dimensional random variable.
    The \textit{hazard rate} of $X$ evaluated at point $x$ is defined by
    \[
        h(x)
        =
        \lim_{\theta \to 0}
        \frac{
            \bbP(
                x \le X \le x + \theta
                |
                X \ge x
            )
        }{
            \theta
        }
        =
        \frac{
            f(x)
        }{
            1 - F(x)
        },
    \]
    where $f$ is the p.d.f (assuming that it exists)
    and $F$ is the c.d.f.
\end{definition}

\begin{assumption}[Additional assumptions on the distribution of noise]
    \label{asm_d}
    We assume that
    for any $\mathbf{u} \in U$,
    we have
    \begin{itemize}
        \item the \textit{hazard rate} of $|w|$ is positive and non-decreasing
        \item the p.d.f of $|w|$ exists and is continuously differentiable.
    \end{itemize}
\end{assumption}
\noindent This assumption is the combination of 
\cite[Assumption~1 and 2]{zhang21}.

\begin{theorem}[Convergence of E-IRPC to a neighborhood of $\mathbf{u}_{PS}$]
    \label{thm_cvg_fs1}
    Under Assumptions~\ref{asm_sc},~\ref{asm_sm},~\ref{asm_lip},~\ref{asm_noise1}, and~\ref{asm_d}, the E-IRPC converges
    to the neighborhood of a unique performative 
    stable control $\mathbf{u}_{PS}$ at a linear rate. Specifically,
    for any $\delta \in (0, 1)$, if we take 
    \[
        {N_i}
        = 
        \cO
        \left(
            \frac{4\lambda^2 T^3}{\beta^2 \delta^2 \sum_{t=0}^{T-1} \eps_t^2}
            \log
            \left(
                \frac{6p}{\pi^2 i^2 T^2}
            \right)
        \right)
    \]
    samples in each iteration, then with probability $1-p$, 
 $|\mathbf{u}_i - \mathbf{u}_{PS}| \le \delta$
    for all  
        $i 
        \ge
        \left ( 
            1 
            - 
            2\alpha_1
        \right )^{-1}
        \log
        \left (
            \frac
            {|\mathbf{u}_0 - \mathbf{u}_{PS}|}
            {\delta}
        \right )$
    with $\alpha_1
        :=
        \frac{\beta \sqrt{\sum_{t=0}^{T-1} \eps_t^2}}{\lambda}$.
\end{theorem}
\begin{proof}
    For each timestep $t$, take the ${N_i}$ samples of $w(t)$ and
    sort their Euclidean norm in the increasing order as follows
        $\left |w^{(\eta_1)}(t) \right | 
        \ge 
        \dots
        \ge
        \left |w^{(\eta_{N_i})}(t) \right |$. 
    Let 
    $
        k 
        = 
        \left \lfloor
            \frac{Np}{T}
        \right \rfloor,
    $
    then the empirical quantile function is given by
    \begin{equation*}
    \begin{aligned}
        \hat{Q}_t(\mathbf{u}) 
        &:= 
        \inf
        \left \{
            r
            :
            \frac{1}{{N_i}}
            \sum_{j = 1}^{N_i}
            \mathbf{1}\left\{
                \left | w^{j}(t) \right |
                \le
                r
            \right\}
            \ge
            1 - \frac{p}{T}
        \right \} \\
        &=
        \left | w^{(\eta_k)}(t) \right |.
    \end{aligned}
    \end{equation*}
    Then for any $\mathbf{u}, \mathbf{u}' \in U$, we have
        $\left |
            \hat{Q}_t(\mathbf{u})
            -
            Q_t(\mathbf{u}')
        \right |
        \le
        \left |
            \hat{Q}_t(\mathbf{u})
            -
            Q_t(\mathbf{u})
        \right |
        +
        \left |
            Q_t(\mathbf{u})
            -
            Q_t(\mathbf{u}')
        \right |$.
    By \cite[Theorem~2]{zhang21}, 
    for any $\eps > 0$, we have
    \begin{equation*}
    \begin{aligned}
        &\qquad \bbP 
        \left(
            \left |
                \hat{Q}_t(\mathbf{u})
                -
                Q_t(\mathbf{u})
            \right |
            \ge
            \eps
        \right) \\
        &\le
        \exp
        \left(
            -\frac
            {\eps^2}
            {
                2 
                \left( 
                    v^r
                    +
                    \left(
                        c^r + \omega_n
                    \right)
                    \eps
                \right)
            }
        \right)
        +
        \exp
        \left(
            -\frac
            {\eps^2}
            {
                2 
                \left( 
                    v^l
                    +
                    \omega_n \eps
                \right)
            }
        \right)
    \end{aligned}
    \end{equation*}
    where 
    $v^r = \frac{2}{kL^2}$,
    $v^l = \frac{2({N_i}-k+1)}{(k-1)^2L^2}$,
    $c^r = \frac{2}{kL}$ and 
    $\omega_n = \frac{b}{{N_i}}$, 
    with $L$ and $b$ finite for all $\mathbf{u} \in U$.
    With some algebraic manipulations, we get that with 
    $
        {N_i}
        = 
        \cO
        \left(
            \frac{4\lambda^2 T^3}{\beta^2 \delta^2 \sum_{t=0}^{T-1} \eps_t^2}
            \log
            \left(
                \frac{6p}{\pi^2 i^2 T^2}
            \right)
        \right)
    $,
    \begin{equation*}
        \bbP 
        \left(
            \left |
                \hat{Q}_t(\mathbf{u}_i)
                -
                Q_t(\mathbf{u}_i)
            \right |
            \ge
            \frac{\beta \sqrt{\sum_{t=0}^{T-1} \eps_t^2}}{\lambda T}
            \delta
        \right)
        \le
        \frac{6p}{\pi^2 i^2 T}.
    \end{equation*}
    Following the same arguments as in the proof of 
    Theorem~\ref{thm_cvg1},
    with probability at least 
    $1 - \frac{6p}{\pi^2 i^2}$, we have
    \begin{equation*}
        |\hat{\cA}(\mathbf{u}_i) - \cA(\mathbf{u}_i)|
        \le
        \frac{\beta}{\lambda}
        \sqrt{
            \sum_{t=0}^{T-1}
            \left|\hat{Q}_t(\mathbf{u}) - Q_t(\mathbf{u}) \right|^2
        }
        \le
        \alpha_1
        \delta.
    \end{equation*}
    Now, applying Theorem~\ref{thm_cvg1} we get that
    when $|\mathbf{u}_i - \mathbf{u}_{PS}| \ge \delta$,
    $
        |\mathbf{u}_{i+1} - \mathbf{u}_{PS}|
        \le
        2 \alpha_1
        |\mathbf{u}_i - \mathbf{u}_{PS}|.
    $
    Therefore, when
    $
        \alpha_1
        < \frac{1}{2}
    $,
    applying the union bound, with probability at least 
    $1 - \sum_{i=0}^{\infty} \frac{6p}{\pi^2 i^2} = 1 - p$,
    for all $i = 0, 1, ...$
    we have
    $|\mathbf{u}_i - \mathbf{u}_{PS}|
        \le
        \max
        \left\{
            \left(
                2\alpha_1
            \right)^i
            |\mathbf{u}_0 - \mathbf{u}_{PS}|
            ,
            \delta
        \right\}$.
    This completes the proof.
\end{proof}

To summarize our analysis, Theorem~\ref{thm_eps} establishes the existance of a performative stable control, which means the I-IRPC algorithm has a fixed point. Theorem~\ref{thm_epo} establishes the existance of a performative optimal control, in which case the problem we investigate is meaningful. Theorems~\ref{thm_cvg1} gives a set of conditions for the I-IRPC algorithm to converge to a performative control, and Theorems~\ref{thm_close1} examines when the corresponding performative control is close to optimal. Finally, Theorems~\ref{thm_cvg_fs1} extends the convergence results in Theorems~\ref{thm_cvg1} to our main method E-IRPC.

\section{Conclusion}
In summary, we present a novel way to approach control tasks with noisy dynamical systems, requiring no prior knowledge of the noise. In particular, we consider settings where the noise depends on our control in unknown ways, which requires us to re-estimate the uncertainty whenever a different control is used. We derive convergence bounds under certain regularity conditions. Note that although we mainly focus on open-loop control in this paper, we believe it is possible to extend our method to closed-loop control, and this will be our primary direction for future research.

\printbibliography

\end{document}